\documentclass{article}
\usepackage{graphicx, amsmath}
\usepackage{amssymb}
\usepackage{amsthm}
\usepackage{mathtools} % Needed for underbraces
\usepackage{thmtools} % Needed for restating lemmas in the appendix

\graphicspath{{Figures/}}

\newtheorem{theorem}{Theorem}
\newtheorem{lemma}{Lemma}
\newtheorem{corollary}{Corollary}
\newcommand{\pn}{\mathop{\textnormal{pn}}}

\renewcommand{\subsection}[1]{\paragraph{\textbf{#1.}}}

\title{Low Ply Drawings of Trees and 2-Trees}

\author{Michael T. Goodrich\thanks{Dept. of Computer Science, University of California, Irvine, {\tt goodrich@uci.edu}}
        \and
        Timothy Johnson\thanks{Dept. of Computer Science, University of California, Irvine, {\tt  tujohnso@uci.edu}}}

\index{Goodrich, Michael T.}
\index{Johnson, Timothy}

\begin{document}
\maketitle

\begin{abstract}
Ply number is a recently developed graph drawing metric inspired by studying road networks. Informally, for each vertex $v$, which is associated with a point in the plane, a  disk is drawn centered on $v$ with a radius that is $\alpha$ times the length of the longest edge incident to $v$, for some constant $\alpha \in (0, 0.5]$. The \emph{ply number} is the maximum number of disks that overlap at a single point. We show that any tree with maximum degree $\Delta$ has a 1-ply drawing when $\alpha = O(1 / \Delta)$. We also show that when $\alpha = 1/2$, trees can be drawn with logarithmic ply number, with an area that is polynomial for bounded-degree trees. Lastly, we show that this logarithmic upper bound does not apply to 2-trees, by giving a lower bound of $\Omega(\sqrt{n / \log n})$ ply for any value of $\alpha$.
\end{abstract}

\section{Introduction}
A useful paradigm for drawing graphs involves visualizing them as maps or road networks, allowing a visualizer to ``zoom'' in and out of the graph based on known techniques that apply to maps. For example, Gansner {\it et al.}~\cite{gansner2010gmap} describe a GMap system for visualizing clusters in graphs as countries with nearby clusters drawn as neighboring countries. In addition, Nachmanson \textit{et al.}~\cite{mondal2017new,nachmanson2015graphmaps} describe a GraphMaps system for visualizing graphs as embedded road networks, so as to leverage the drawing and zooming capabilities of a roadmap viewer to explore the graph. Thus, a natural question arises as to which graphs are amenable to being drawn as road networks.

To answer this question, we formulate a precise definition of what we mean by a graph that could be drawn as a road network. One might at first suggest that graph planarity would be a good choice for such a formalism. But the class of planar graphs includes several graph instances that are difficult to visualize as road networks, such as the so-called ``nested triangles'' graph (e.g., see~\cite{JGAA-251,Frati2008,Garg1994}). In addition, as shown by Eppstein and Goodrich \cite{eppstein2008studying}, the class of planar graphs is not general enough to include all real-world road networks, as road networks are often not planar. For example, the California highway system alone has over 6,000 crossings. Instead of using planarity, then, Eppstein and Goodrich \cite{eppstein2008studying} introduce the concept of the \emph{ply number} of an embedded graph, and they demonstrate experimentally that real-world road networks tend to have small ply. Intuitively, the ply concept tries to capture how road networks have features that are well-separated at multiple scales. The formal definition of the ply number of a graph is derived from the definition of ply for a set of disks (which captures the depth of coverage for such a set of disks)~\cite{Miller-1137}; hence, the ply number of an embedded graph is defined in terms of the ply of a set of disks defined with respect to this embedding.

Let us therefore formally define the \emph{ply number} of an embedded geometric graph. Let $\Gamma$ be a straight-line drawing of a graph $G$. For every vertex $v \in G$, let $C_v^{\alpha}$ be the open disk centered at $v$ and whose radius $r_v^{\alpha}$ is $\alpha$ times the length of the longest edge incident to $v$. The set of ply disks containing a point $q$ is then $S_q^{\alpha} = \{C_v^{\alpha} \ | \  \| v - q \| < r_v^{\alpha} \}$. The \emph{$\alpha$-ply number} of this drawing is defined as
\[ \pn(\Gamma) = \max_{q \in \mathbb{R}^2} \| S_q^{\alpha} \| . \]
Usually, $\alpha$ is chosen in the range $(0, 0.5]$. In this range, a graph with two vertices and a single edge connecting them has ply number 1, because the ply disks for the two vertices will not overlap.

There are two natural optimization problems for the ply number when constructing a drawing of a given graph. One is to is to fix a constant ply number, and try to find a drawing that maximizes the value of $\alpha$. The other is to fix a value for $\alpha$, typically $1/2$, and try to find a drawing that minimizes the ply number. In this paper, we provide new results for both of these cases.

\subsection{Previous related work}
As an empirical justification of the use of ply numbers, De~Luca {\it et al.}'s experimental study~\cite{de2017experimental} found that some force-directed algorithms, including Kamada-Kawai \cite{kamada1989algorithm}, stress majorization \cite{gansner2004graph}, and the fast multipole method \cite{hachul2004drawing} all tend to produce drawings with low ply number. Their experiments also suggest that even trees with at most three children per node can have unbounded ply number when $\alpha = 1/2$.

The problem of drawing graphs with ply number equal to 1 is related to that of constructing circle-contact representations. A circle-contact representation for a graph is a collection of interior-disjoint circles, in which each circle represents a single vertex, and two vertices are adjacent if and only if their circles are tangent to one another \cite{hlinveny1995contact, hlinveny1998classes}. Di Giacomo {\it et al.}~\cite{di2015low} show that graphs with ply number 1 are equivalent to graphs with weak unit disk contact representations, which are known to be NP-hard to recognize~\cite{breu1998unit}. They also show that binary trees have drawings with ply number 2 when $\alpha = 1/2$, or with ply number 1 when $\alpha \leq 1/3$. Their drawing is reproduced in Figure~\ref{fig:alpha-ply-binary}.

\begin{figure} 
\centering
\includegraphics[scale=0.4]{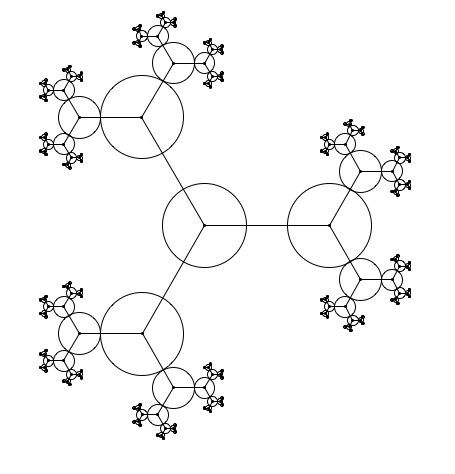}
\caption{Di Giacomo {\it et al.}'s drawing of a binary tree with $\alpha$-ply number 1, for $\alpha = \frac{1}{3}$. The edge lengths decrease by a factor of 2 at each level.}
\label{fig:alpha-ply-binary}
\end{figure}

Angelini {\it et al.}~\cite{angelini2017vertex} relax our definition of ply number to define the \textit{vertex-ply} of a drawing, which is the maximum number of intersecting disks at any vertex of the drawing. Graphs with vertex-ply number 1 can then be interpreted as a new variant of proximity drawings. 

In an earlier paper, Angelini {\it et al.}~\cite{angelini2016low} show that 10-ary trees have unbounded ply number. Furthermore, they prove that 5-ary trees can be drawn with logarithmic ply number and polynomial area. The ply number of drawings of trees with between three and nine children per node remains an interesting and surprisingly daunting open problem.

\subsection{Our results}
In this paper, we study a number of related problems concerning low-ply drawings of bounded-degree trees. We first answer an open question proposed by Di~Giacomo {\it et al.}~\cite{di2015low}, which asks whether all trees with maximum degree $\Delta$ have 1-ply drawings for a sufficiently small $\alpha$. We show in Section~2 that a simple fractal drawing pattern can achieve this when $\alpha = O(1 / \Delta)$.

In Section~3, we show that all trees (not just 5-ary trees) can be drawn with logarithmic ply number, for $\alpha = 1/2$. Furthermore, the area is polynomial for trees with bounded degree. These results depend on some careful arguments about geometric configurations and fractal-like geometric constructions, as well as yet another use of the heavy-path decomposition technique of Sleator and Tarjan~\cite{sleator1983data}.

It is then natural to consider whether any planar graph classes larger than trees can be drawn with logarithmic ply number. In Section~4, we show that this is not the case for $2$-trees, by constructing a family of $2$-trees that require a ply number of $\Omega(\sqrt{n / \log n})$ for any fixed $\alpha > 0$. Previous lower bounds have only applied for planar drawings, while non-planar drawings are known to sometimes have better ply number.

\section{1-ply Drawings}
In this section, we fix our drawings to have ply number 1, and provide conditions on $\alpha$ such that we can construct drawings of trees of any bounded degree. At a high level, our drawings are constructed as follows. For a tree with maximum degree $\Delta$, we divide the area around each parent vertex radially into $\Delta$ equal wedges, so that all of the angles are $2 \pi / \Delta$. Then we assign each subtree to a different wedge, and draw it within that wedge. The distance from each node to its children is chosen to be a constant fraction $f$ of its distance from its own parent. When $\Delta = 3$, this produces the drawing shown in Figure~\ref{fig:alpha-ply-binary}. Note that for a non-root vertex, one of the wedges will contain the edge from the parent vertex, and will not contain a subtree.

This produces a drawing that is highly symmetric, in a fashion that would produce a fractal if continued in the limit.\footnote{See Falconer~\cite{falconer2004fractal} for further reading about fractal geometry.} Thus, any bounded-degree tree is a subtree of this infinite tree; hence, this drawing algorithm can produce a drawing of any bounded-degree tree. Filling in the details of this construction requires setting the values of two parameters: $f$, the ratio between outgoing and incoming edge lengths; and $\alpha$, the ratio between the radius of a ply disk for a vertex and the length of its longest incident edge. We provide constraints for the following three cases, which taken together ensure that there are no overlaps, so that the ply number of our drawings is 1. We then maximize $\alpha$ such that all of these constraints are satisfied.  
\begin{enumerate} 
\item Ply disks for adjacent vertices must not overlap.  
\item Ply disks for vertices in separate subtrees must not overlap.  
\item A ply disk for a vertex must never overlap a ply disk for one of its descendants.
\end{enumerate}
It is easily verified that these three conditions are necessary and sufficient for a tree to have a 1-ply drawing.

\subsection{Condition 1: Separate adjacent vertices}
Except for the root vertex, which has no incoming edge, we proportion the lengths of the edges for each vertex as shown in Figure~\ref{fig:adjacent}.

\begin{figure}
\centering
\includegraphics[scale = 0.7]{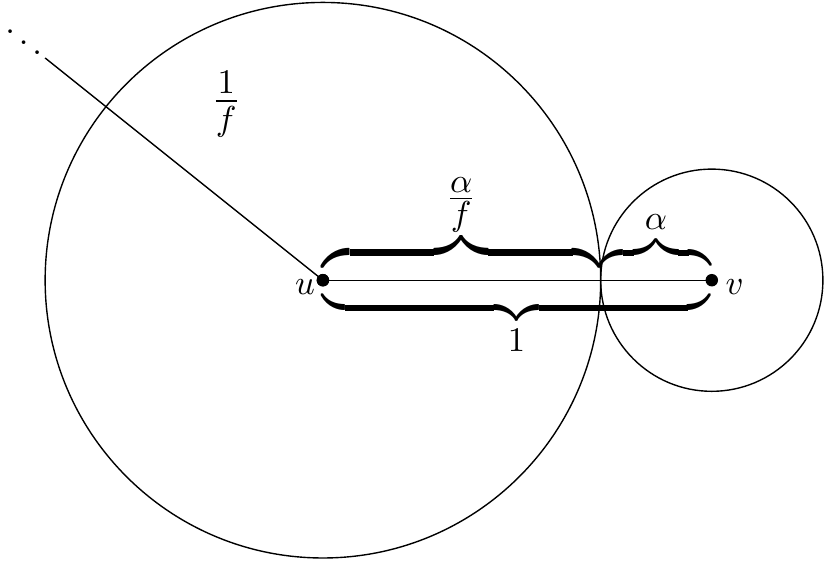}
\caption{Our edges decrease by a factor of $f$ at each level, and the
ply disks have radius $\alpha$ times the length of the incoming edge.}
\label{fig:adjacent}
\end{figure}

That is, taking the length of the reference edge $\overline{uv}$ as 1 (illustrated in Figure~\ref{fig:adjacent} going from parent to child in a left-to-right orientation), then, based on our definition of the $\alpha$-ply number, the radius of the larger circle is $\alpha/f$, the radius of the smaller circle is $\alpha$, and their distance is 1. Thus, we have our first condition relating $\alpha$ and $f$.
\begin{align*}
&\alpha / f + \alpha \leq 1 \\
&\alpha \leq \frac{f}{1 + f}
\end{align*}

\subsection{Condition 2: Separate subtrees with the same root}
We require that the ply disks for any subtree all be contained within a wedge of angle $\theta=\frac{2 \pi}{\Delta}$ around its parent vertex, where $\Delta$ is the degree. Since our wedges for each subtree are disjoint, this ensures that the ply disks for two adjacent subtrees cannot overlap.

As illustrated in Figure~\ref{fig:wedge}, the distance from a child vertex to the boundary of its containing wedge is $d = \sin \left( \frac{\pi}{ \Delta } \right)$. Note also that the lengths of edges along a path in this subtree form a geometric sequence with ratio $f$. So the maximum distance from a child vertex to any vertex in its subtree is $\sum_{i = 1}^{\infty} f^i = \frac{f}{1 - f}$.

\begin{figure}
\centering
\includegraphics[scale = 0.7]{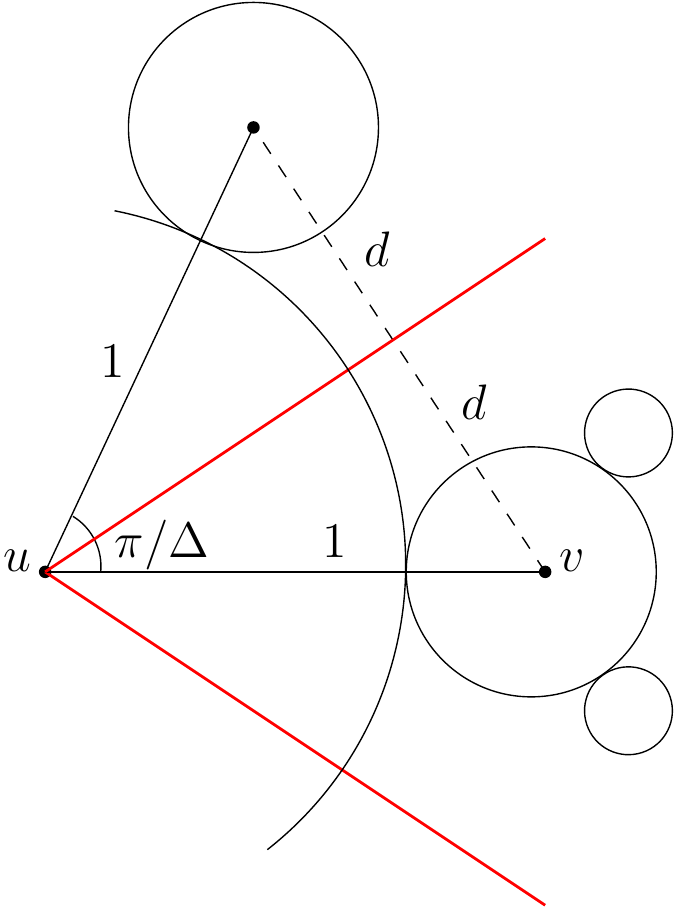}
\caption{Our constraint on a wedge containing a subtree of a central vertex.}
\label{fig:wedge}
\end{figure}

Therefore, to confine each subtree within its wedge, we must set
\[ \frac{f}{1 - f} \leq \sin \left( \frac{\pi} { \Delta } \right). \]
Solving for $f$, we get
\begin{align}
f \leq \frac{ \sin \left( \frac{\pi} {\Delta} \right) }{1 + \sin
\left( \frac{\pi} {\Delta} \right) }.
\end{align}

\subsection{Condition 3: Separate each vertex from its descendants}
Our last condition is that the ply disk for a vertex cannot overlap any of its descendants. The closest descendants will be those in the wedges on either side of the edge between their parent and grandparent, which are at an angle of $\frac{2 \pi}{\Delta}$ from their parent, as in Figure~\ref{fig:ancestor}.

\begin{figure}
\centering
\includegraphics[scale=0.9]{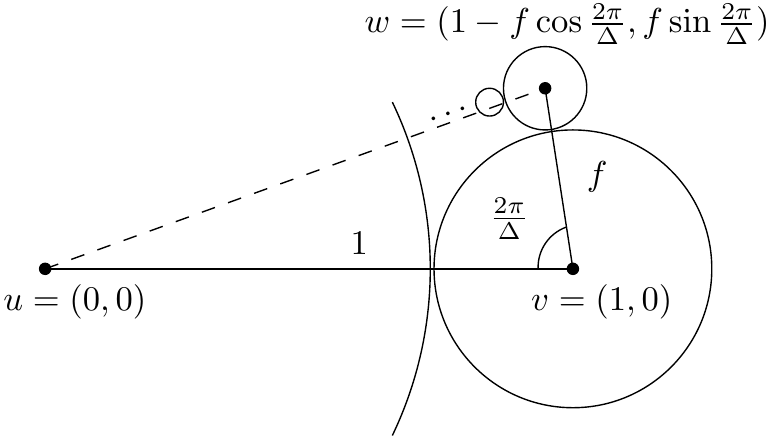}
\caption{Our layout leaves a gap of angle $\theta=\frac{2 \pi}{\Delta}$ for the edge from the parent vertex. The descendants on either side must not overlap their ancestors.}
\label{fig:ancestor}
\end{figure}

Once again we normalize $(u,v)$, as having length~1. We then perform a rigid transformation that
takes the grandparent, $u$, to the origin so that the edge $(u,v)$ is along the $x$-axis, $u$'s closest grandchild, which we call $w$, is located at  the point $\left( 1 - f \cos \left( \frac{2 \pi} {\Delta} \right), f \sin \left( \frac{2 \pi} {\Delta} \right) \right)$.  We require that the distance from $w$ to its descendants be no greater than the distance from $w$ to the boundary of the ply disk for $u$. Recall that our wedge angle $\theta = \frac{2 \pi} {\Delta}$.  We apply the following constraint:
\[\sqrt{(1 - f \cos \theta)^2 + (f \sin \theta)^2} \geq \frac{\alpha}{f} + \sum_{i = 2}^{\infty} f^i \]

After simplifying and solving for $\alpha$, our condition is
\[\alpha \leq f \sqrt{1 - 2f \cos \theta + f^2} - \frac{f^3}{1 - f}\]

Let us now compare our three conditions. We see that equation 2 gives us an upper bound for $f$, while equations 1 and 3 give us upper bounds for $\alpha$ that both increase as $f$ gets larger. So to maximize $\alpha$, we let $f$ be equal to its upper bound. This gives us the following theorem.

\begin{theorem}
Let $T$ be a tree with maximum degree $\Delta$, and let
\[f = \frac{ \sin \left( \frac{\pi} {\Delta} \right) }{1 + \sin
\left( \frac{\pi} {\Delta} \right)}.\]

$T$ has a 1-ply drawing if
\[\alpha \leq \min \left( \frac{f}{1 + f}, \ f \sqrt{1 - 2f \cos (2 \pi / \Delta) + f^2} - \frac{f^3}{1 - f} \right).\]
\end{theorem}

\begin{corollary}
A tree with maximum degree $\Delta$ has a 1 ply drawing when $\alpha~=~\Theta(1 / \Delta)$.
\end{corollary}

\begin{proof}
First, recall that we defined:
\[f = \frac{ \sin \left( \frac{\pi} {\Delta} \right) }{1 + \sin \left( \frac{\pi} {\Delta} \right)}.\]

Now we will consider the limiting value of $\Delta \cdot f$.
\[\lim_{\Delta \rightarrow \infty} \Delta \cdot f = \lim_{\Delta \rightarrow \infty} \frac{\Delta \sin (\pi / \Delta)}{1 + \sin (\pi / \Delta)} = \pi\]

Therefore, $f = \Theta (1 / \Delta)$. So as $\Delta \rightarrow \infty$, $f \rightarrow 0$. \newline

Secondly, recall that in our theorem we showed:
\[\alpha \leq \min \left( \frac{f}{1 + f}, \ f \sqrt{1 - 2f \cos (2 \pi / \Delta) + f^2} - \frac{f^3}{1 - f} \right)\]

Suppose that we use the first condition, $\alpha = f / (1 + f)$. Then $\alpha / f = 1 / (1 + f)$. So $\lim_{f \rightarrow 0} \alpha / f = 1$. \newline

Then suppose that we use the second condition:
\[\alpha = f \sqrt{1 - 2f \cos (2 \pi / \Delta) + f^2} - \frac{f^3}{1 - f}\]

Again, $\lim_{f \rightarrow 0} \alpha / f = 1$, so $\alpha = \Theta(f) = \Theta(1 / \Delta).$
\end{proof}

Note, however, that some of our conditions are not tight. For condition 2, we assumed that the branches of our subtrees would approach the sides of their wedge directly. But when the degree of our tree is 4, the angle between two subtrees is $90^{\circ}$. Therefore, every edge in our tree is either horizontal or vertical, so we can measure the distance to the boundary of the wedge using Manhattan distance instead of Euclidean distance. (See Figure~\ref{fig:wedge-manhattan}.)

So for a tree with degree 4, we replace condition 2 with the following equation:
\[\sum \limits_{i = 1}^{\infty} f^i \leq 1\]
This implies $f = 1/2$, and our other conditions imply $\alpha = 1/3$. In this case, our bound is tight. (See Figure~\ref{one-ply-example}.)

\begin{figure}
\centering
\includegraphics[scale=1]{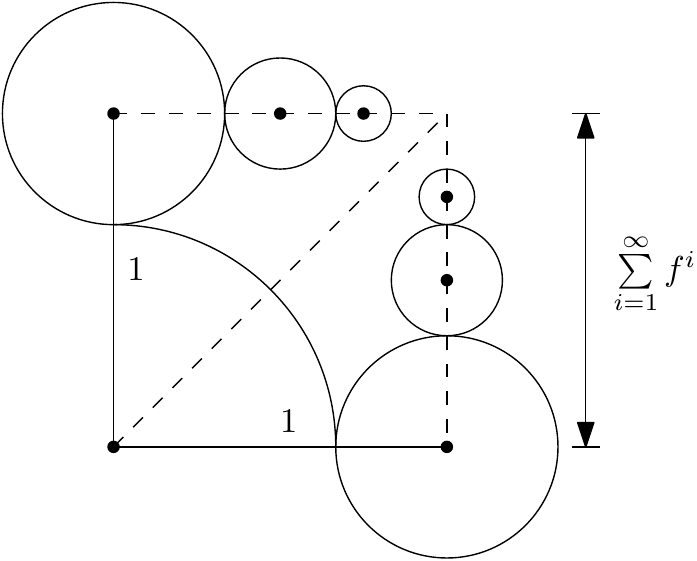}
\caption{An improved bound for Condition 2. The Manhattan distance is sufficient to confine subtrees within a wedge when all edges are either horizontal or vertical.}
\label{fig:wedge-manhattan}
\end{figure}

\begin{figure}
\centering
\includegraphics[scale=0.5]{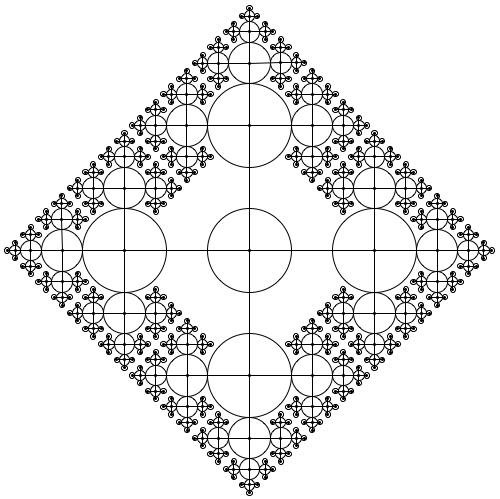}
\caption{A 1-ply drawing of a tree with maximum degree four, for which $f = 1/2$, $\alpha = 1/3$.}
\label{one-ply-example}
\end{figure}

\section{Polynomial area, logarithmic ply number}
In this section, we prove the following theorem.

\begin{theorem} \label{thm:log-drawing}
For $\alpha = 0.5$, a tree with maximum degree $\Delta$ can be drawn with ply number $O(\log n)$ in area $n^{O(\Delta)}$.
\end{theorem}

Note that for a bounded-degree tree, $\Delta$ is a constant, so our area is polynomial in $n$. We first give a simple fractal layering algorithm that proves our theorem for balanced trees. Then we extend it to all trees by using a heavy path decomposition. A similar approach was used by Angelini \textit{et. al.}~\cite{angelini2016low} for drawing trees up to maximum degree six, but we add our layering technique to make their algorithm work for all trees.

\subsection{Radially layered drawings}
We begin with a simple algorithm for drawing trees by layering their children. For each vertex, we choose a sequence of distances $d_i$ for the layers, such that vertices in adjacent layers have disjoint ply disks. 

\begin{lemma} 
Suppose that $r$ is the root of a star graph. Let $v_1, v_2$ be children at distances $d_1, d_2$, respectively. If $d_2 \geq 3d_1$, then the ply disks for $v_1$ and $v_2$ are disjoint.
\label{lemma:layer-size}
\end{lemma}

\begin{proof}
The distance to $v_1$ is $d_1$, so since $\alpha = 0.5$, its ply disk will have radius $0.5 d_1$, and will be contained within an open disk of radius $1.5 d_1$ centered at $r$. The distance to $v_2$ is $d_2$, so its ply disk will have radius $0.5 d_2$. Its closest approach to $r$ will be at distance $0.5 d_2 \geq 1.5 d_1$. Thus, the ply disks for $v_1$ and $v_2$ are disjoint. (See Figure \ref{layer-size}.)
\end{proof}

\begin{figure}
\centering
\includegraphics[scale=0.7]{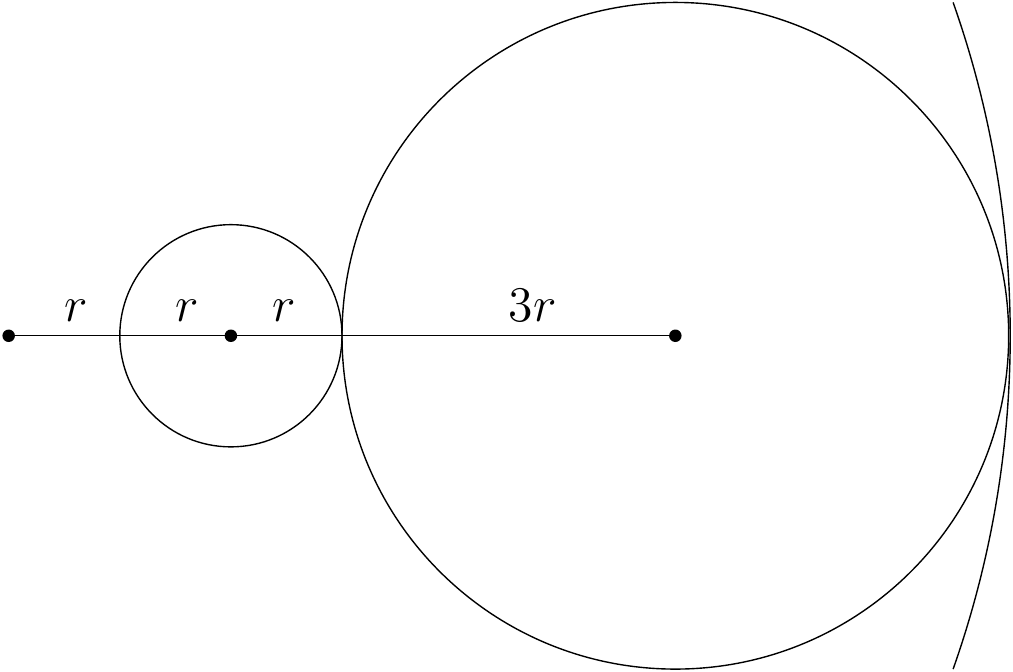}
\caption{If each layer in a tree drawing is at least three times as far as the previous layer, the ply disks for the layers will not overlap. In this figure, $d_1 = 2r$ and $d_2 = 6r$, so our condition holds.}
\label{layer-size}
\end{figure}

Next, note that we can put up to six vertices in each layer without overlaps. So for a tree with degree $\Delta$, we need $\lceil \Delta/6 \rceil$ layers. We pick any desired size for the initial layer around our root, then draw the subtrees for each child vertex recursively within their own ply disks. Therefore, the size of the smallest layer must shrink by a factor of $3^{\lceil \Delta/6 \rceil}$ each time we add a level to our tree.

Since our tree is balanced, its total height is $O(\log n)$. Thus the ratio of the longest to the smallest edge is $3^{O(\Delta \log n)} = n^{O(\Delta)}$. The area will then also be $n^{O(\Delta)}$, for a larger constant.

This completes our proof for balanced trees. Figure \ref{fig:layered-tree} provides an example drawing of such a tree with degree 18 using three layers.

\begin{figure} 
\centering
\includegraphics[scale=0.4]{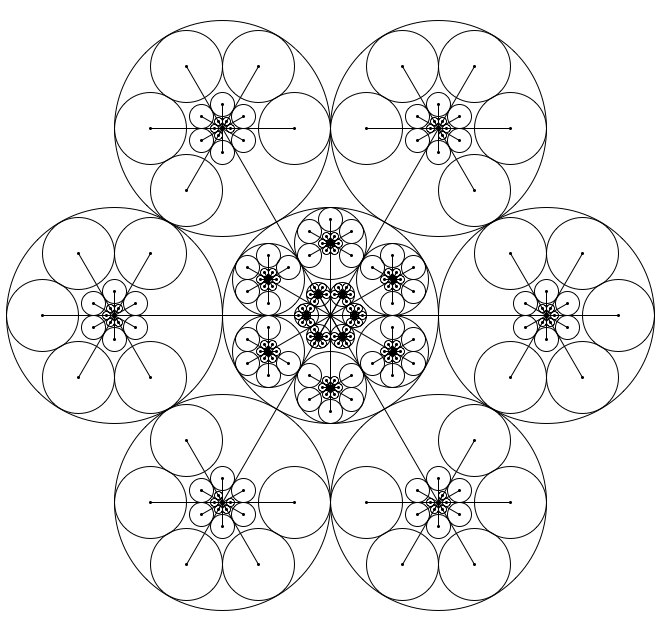}
\caption{A tree with degree 18, where the children of each vertex are drawn in three layers.}
\label{fig:layered-tree}
\end{figure}

\subsection{Heavy path decomposition}
When our trees are not balanced, we will use the heavy path decomposition~\cite{sleator1983data} to still produce drawings with logarithmic ply number. This decomposition partitions the vertices in our tree into paths that each end at a leaf. To choose the first path, we begin at the root. Then from its child subtrees, we choose the largest one and add its root to our path. We continue downward until we reach a leaf.

We next remove the vertices on this path from our tree, creating a new set of subtrees, and repeat the same process for each subtree. That is, the root vertex for each of these subtrees will become the starting point for a new path constructed by the same process. We recurse until every vertex in our tree is assigned to some path. The subtrees that are rooted at a child of a vertex $v$ and whose root is not on the same path as $v$ are are said to be \textit{anchored} at $v$. The path containing the root of each of those subtrees is also said to be anchored at $v$.

The set of paths constructed by this process now itself forms a new tree (see Figure \ref{fig:heavy-path}), in which the path $P_i$ is a parent of $P_j$ if one of the vertices in $P_i$ is an anchor for $P_j$. We will show that the ply number of our drawings is proportional to the height of this decomposition tree, which is known to be $O(\log n)$.

Now we describe how to draw each path in the decomposition tree. First, we define a \textit{2-drawing} of a path $P = (v_1, \ldots, v_m)$ as a straight-line drawing of $P$ along a single segment that satisfies the following properties.

\begin{itemize}
\item All of the vertices appear in the line segment in the same order as they appear in $P$.
\item For each $i = 2, \ldots, m - 1$ we have  $\frac{l(v_{i - 1}, v_i)}{2} \leq l(v_i, v_{i + 1}) \leq 2 l(v_{i - 1}, v_i)$.
\end{itemize}

\begin{lemma}
A 2-drawing of a path has ply number at most 2.
\end{lemma}

\begin{proof}
See Lemma 5 in Angelini {\it et al.}~\cite{angelini2016low}.
\end{proof}

\begin{figure}
\centering
\includegraphics[scale=0.7]{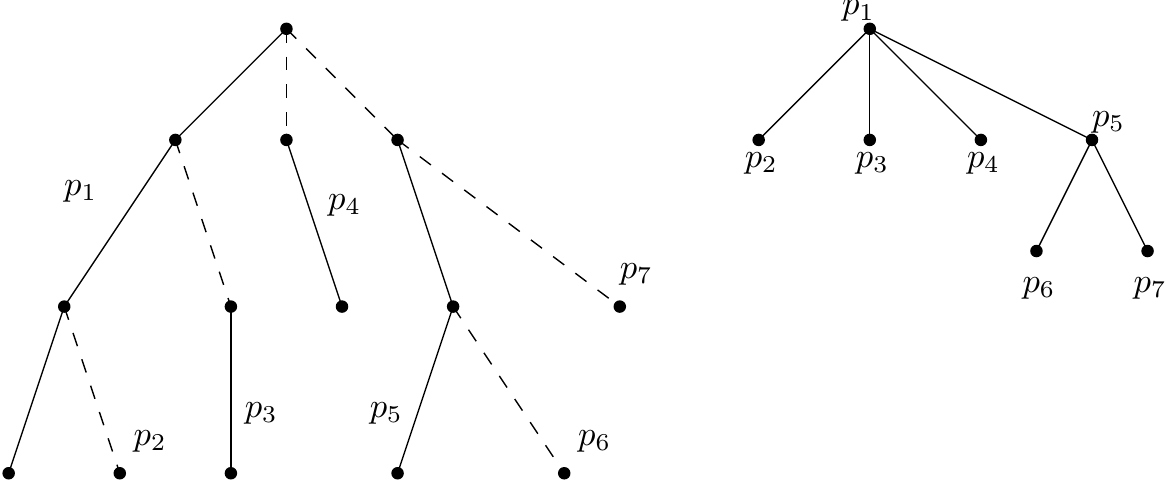}
\caption{A tree and its heavy path decomposition.}
\label{fig:heavy-path}
\end{figure}

\begin{figure}
\centering
\includegraphics[scale=0.65]{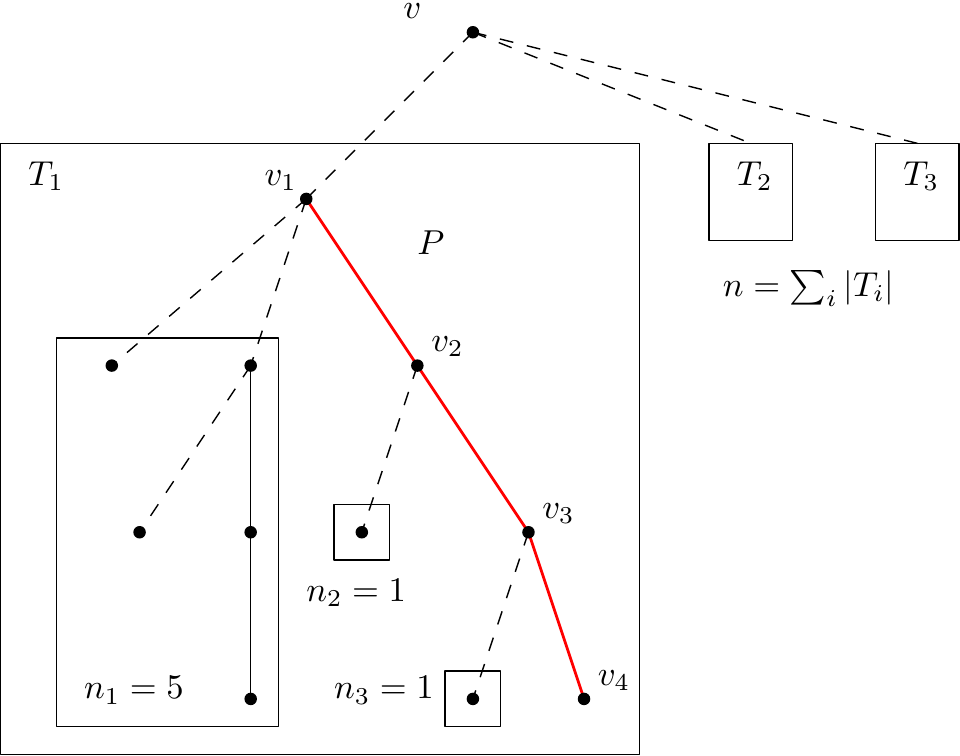}
\caption{Labels for different sizes in a heavy path decomposition tree.}
\label{fig:heavy-path-labels}
\end{figure}

Now suppose that we have a path $P = (v_1, v_2, \ldots, v_k)$ in our heavy path decomposition, and let $P$ be anchored at vertex $v$, so that $v$ is the parent of $v_1$. Let $n$ be the total size of the subtrees anchored at $v$, and let $n_i$ be the total size of the subtrees anchored at $v_i$ (Figure~\ref{fig:heavy-path-labels}). Lastly, we denote the length of the edge $(v, w)$ as $l(v, w)$.

Intuitively, we want to draw each path so that more space is available for vertices that have larger subtrees. At the same time, we want to ensure that the lengths of the two edges for a vertex are within a factor of two, so that our path is a 2-drawing. This can be achieved using the following algorithm \textsc{DrawPath}.

To draw the path $P$, we first set $l(v, v_1) = n_1$ and $l(v_i, v_{i + 1}) = n_i + n_{i + 1}$, for each $i = 1, \ldots, k - 1$. Next we visit the edges of our path in decreasing order of length. When an edge $(v_i, v_{i + 1})$ is visited, we make sure that both of its neighboring edges are at least half as long. That is, we set:
\begin{itemize}
\item $l(v_{i - 1}, v_i) = \max \{ \frac{l(v_i, v_{i + 1}}{2}, l(v_{i - 1}, v_i) \}$
\item $l(v_{i + 1}, v_{i + 2}) = \max \{ \frac{l(v_i, v_{i + 1}}{2}, l(v_{i + 1}, v_{i + 2})\}$
\end{itemize}

\begin{lemma} \label{lem:two-drawing-length}
The algorithm \textsc{DrawPath} constructs a 2-drawing $\Gamma$ of $P$ such that $l(v, v_1) \geq n_1$, $l(v_i, v_{i + 1}) \geq n_i + n_{i + 1}$, and for each $i = 1, \ldots, m - 1$, and $l(P) \leq 6 n$.
\end{lemma}

\begin{proof}
See Lemma 6 in Angelini {\it et al.}~\cite{angelini2016low}.
\end{proof}

\begin{figure}
\centering
\includegraphics[scale=0.6]{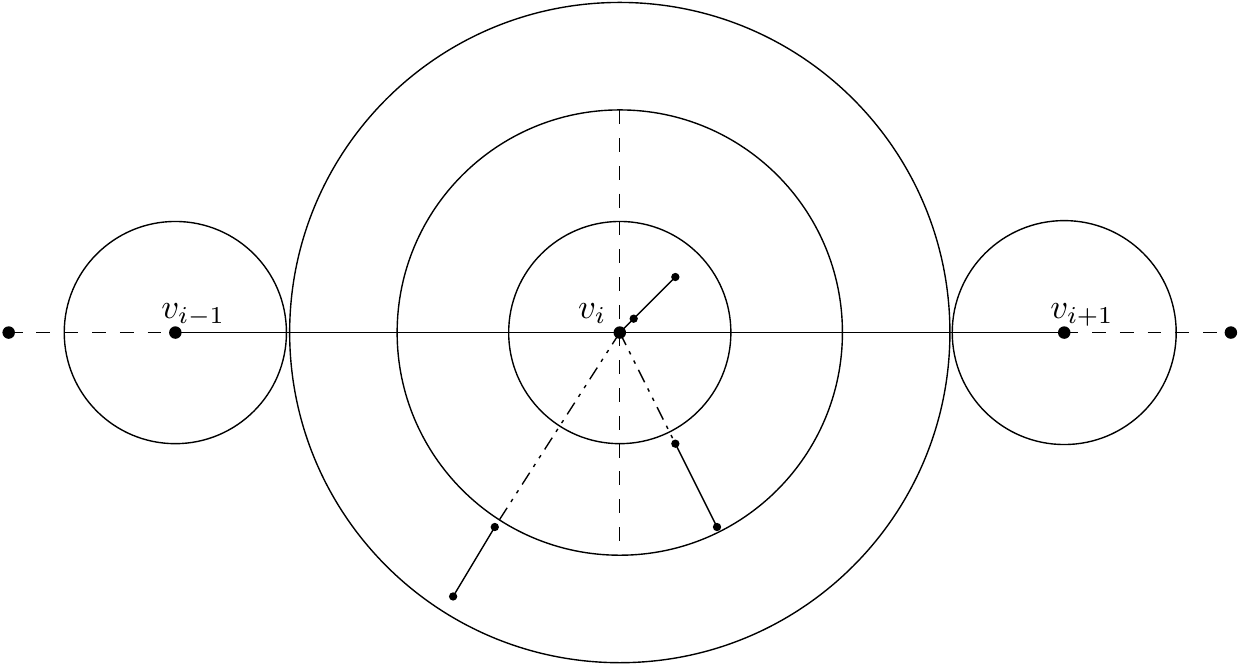}
\caption{Three vertices along a path in our decomposition, along with their drawing disks (not the ply disks). For the center vertex $v_i$, we show three paths in different layers around it, which would be drawn recursively.
}
\label{fig:layered-paths}
\end{figure}

We now perform a bottom-up construction of our tree, drawing each path using the \textsc{DrawPath} algorithm. Once all of the paths anchored at vertices in $P$ have been drawn, we construct a drawing of $P$ with each path in a separate layer (Figure~\ref{fig:layered-paths}). This translation may increase the ply radius of the first vertex in each of these paths, so the ply number of the drawing for each path may increase from 2 to 3.

In the appendix, we prove the following properties of this drawing.

\begin{restatable}{lemma}{drawinglemma} \label{lem:drawinglemma}
For each vertex $v$ we can associate a drawing disk $D_v$ (which is distinct from the ply disk for $v$) that satisfies the following properties.
\begin{enumerate}
\item If $v, w$ are two distinct vertices on the same path, then their disks $D_v, D_w$ are disjoint.
\item The ply disks for the subtrees anchored at $v$ are all contained within $D_v$, and are within disjoint layers.
\item Each path is scaled by a factor of $O(3^{\Delta})$ larger than the paths that are anchored at its vertices.
\end{enumerate}
\end{restatable}

Together, these properties imply that the ply disks for a path can only overlap with ply disks for their ancestor paths in the heavy path decomposition tree. Therefore, since each path is drawn with ply number at most 3, the total ply number is at most $3 (h + 1)$, where $h$ is the height of the heavy path decomposition tree. Since $h = O(\log n)$, the ply number is $O(\log n)$.

Lastly, if $\Delta$ is a constant, then the total scaling for our largest disk is $3^{O(\Delta \log n)}$, which simplifies to $n^{O(\Delta)}$. This completes our proof of Theorem~\ref{thm:log-drawing}.

\section{Lower bound for 2-trees}
Since all trees can be drawn with $O(\log n)$ ply number, it is natural to consider larger planar graph classes. We show that a 2-tree can require at least $\Omega(\sqrt{n / \log n})$ ply, for any fixed $\alpha$.

First, let us informally describe a 2-tree. A tree can be constructed by beginning with a root vertex, and adding vertices one a time, attaching each new vertex to a single parent. A 2-tree can be constructed by beginning with two root vertices connected by an edge. Then each time we add a new vertex, we attach it to two different parents.

We know that a star can be drawn with ply number 2 when the distance to successive vertices increases exponentially ~\cite{angelini2016low}. A tree can be drawn with $O(\log n)$ ply number when the distances from parents to their children decrease exponentially as we move down the tree. Intuitively, combining these two graphs produces a graph that requires large ply, since it is impossible to satisfy both conditions simultaneously.

Accordingly, we begin with $m$ disjoint complete binary trees of height $h$, which we label $T_i$, $1 \leq i \leq m$, where $m$ and $h$ will be determined later. Then we add one vertex $v$ connected to every vertex in each tree. Note this graph can be constructed as a subgraph of a 2-tree.

Now we have two possible types of drawings for our graph. In one case, every tree has some vertex whose ply disk contains $v$. Therefore, the ply number of our graph is at least $m$, since there are at least $m$ ply disks that all contain $v$. In the second case, there is some tree $T_i$ for which none of the ply disks for its vertices contain $v$. To analyze this case, we make use of the following lemma.

\begin{lemma}
If two adjacent vertices $w_1$ and $w_2$ satisfy $d(v, w_1) > (1 + 1 / \alpha) d(v, w_2)$, then the ply disk for $w_2$ contains $v$.
\end{lemma}

\begin{proof}
Suppose that $d(v, w_2) = r$, so that $d(v, w_1) > (1 + 1 / \alpha)r$. By the triangle inequality, $d(w_1, w_2) > r / \alpha$. Then the ply radius for $w_1$ is at least $\alpha d(w_1, w_2) > r$. Therefore, the ply disk for $w_1$ contains $v$. (See Figure~\ref{fig:triangle}.)
\end{proof}

\begin{figure}
\includegraphics[scale=0.7]{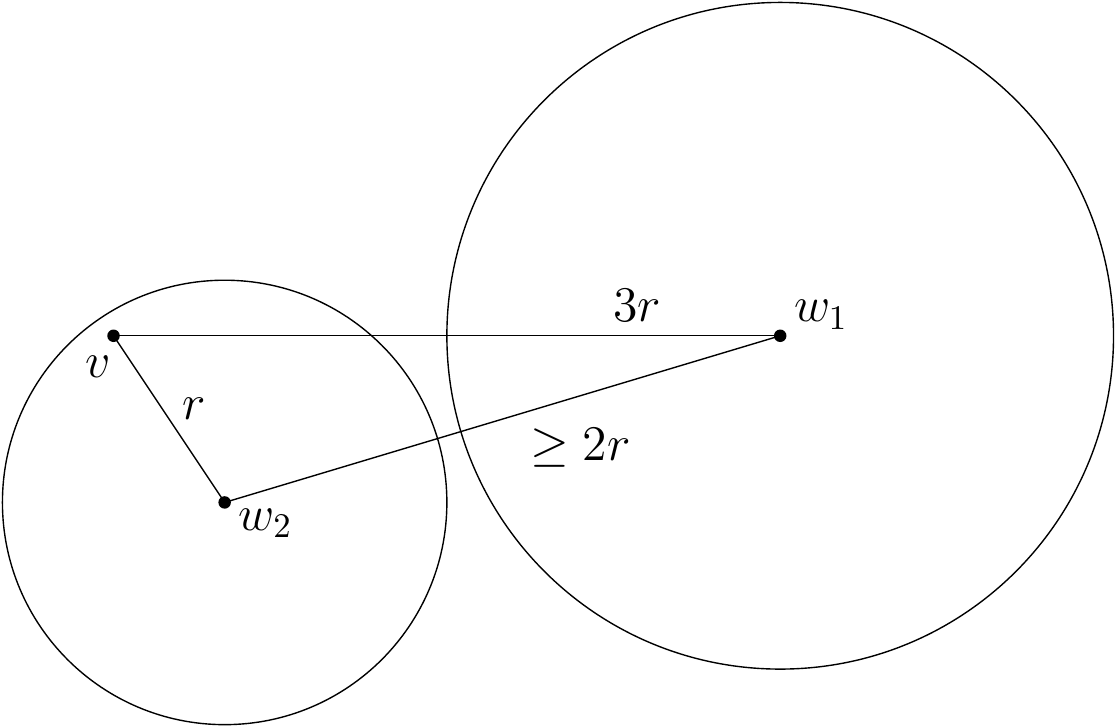}
\caption{If two adjacent vertices $w_1$ and $w_2$ are both attached to $v$, and $w_1$ is much further from $v$, then the ply disk for $w_2$ will contain $v$.}
\label{fig:triangle}
\end{figure}

Assume without loss of generality that the distance from $v$ to the root of $T_i$ is 1, and let $c = 1 + 1 / \alpha$. We can then show by induction that if no ply disk in $T_i$ contains $v$, then the nodes at the $j$th level of our tree are at distance at most $c^j$ from $v$, and at least $c^{-j}$.

Now partition our drawing into annulae $S_l$, where the inner radius of $S_l$ is $c^l$, and the outer radius is $c^{l + 1}$, for $-h \leq l \leq h - 1$. Next choose $\bar{l}$ to be the index of the annulus containing the maximum number of vertices. We have more than $2^h$ vertices, and only $2h$ annulae to distribute our vertices, so $S_{\bar{l}}$ must contain at least $2^h / 2h$ vertices. Since each of these vertices is at a distance of at least $c^{\bar{l}}$, each has a ply radius of at least $c^{\bar{l}} \alpha$.

A vertex at a distance of $c^{\bar{l} + 1}$ from $v$ will have a ply radius of at least $\alpha c^{\bar{l} + 1}$, so the outer edge of its ply disk will be at a distance of $(\alpha + 1) c^{\bar{l} + 1}$. Therefore, let $D$ be the disk centered at $v$ with a radius of $r_D = (\alpha + 1) c^{\bar{l} + 1}$, so that all of the ply disks for vertices in $S_{\bar{l}}$ are contained in $D$. Now we compute the ratio of the areas of the ply disks in $D$ to its own area, which is a lower bound for the ply number. Note that $D$ contains at least $2^h / 2h$ ply disks that each have a radius of at least $c^{\bar{l}} \alpha$. Therefore, this ratio is at least:
\[\underbrace{\frac{2^h}{2h}}_{\mathclap{\text{vertices}}} \overbrace{(\pi \alpha^2 c^{2\bar{l}})}^{\mathclap{\text{ply area per vertex}}} \underbrace{\frac{1}{\pi (\alpha + 1)^2 c^{2 \bar{l} + 2}}}_{\mathclap{\text{inverse disk area}}} = \frac{2^h}{h} \frac{\alpha^2}{(\alpha + 1)^2 c^2} = \Omega(2^h / h)\]

Now let $h = (\log n + \log \log n) / 2$, and let $m = \sqrt{n / \log n}$. Note that the total number of vertices in each tree is $2^{(\log n + \log \log n) / 2} = \sqrt{n \log n}$. The total number of vertices overall is then $m \cdot (2^{h + 1} - 1) + 1 = O(n)$.

If every tree $T_i$ has a vertex whose ply disk contains $v$, then the ply number is at least $m = \sqrt{n / \log n}$. Otherwise, if some tree does not have such a vertex, then that tree's ply number is $\Omega(2^h / h) = \Omega(\sqrt{n / \log n})$. This gives us the following theorem.

\begin{theorem}
There is a 2-tree with $O(n)$ vertices for which any drawing has ply number $\Omega(\sqrt{n / \log n})$, for any fixed $\alpha > 0$.
\end{theorem}

\section{Conclusion}
We have shown that all trees have 1-ply drawings when $\alpha = O(1 / \Delta)$, or logarithmic ply number when $\alpha = 0.5$, and that 2-trees may require $\Omega(\sqrt{n / \log n})$ ply for any $\alpha$.

There are many open questions left to resolve, but we are especially interested in closing the gap between constant and logarithmic ply for trees with between three and nine children per node. We would also like to consider intermediate planar graph classes between trees and 2-trees, such as outerplanar graphs, and determine whether they can be drawn with $O(\log n)$ ply.

\section*{Acknowledgements}
This research was supported by 
DARPA agreement no.~AFRL FA8750-15-2-0092
and NSF grants 1526631 and 1815073.
The views expressed are those of the authors and do not reflect the 
official policy or position of the Department of Defense 
or the U.S.~Government.

\clearpage
\bibliographystyle{plain}
\bibliography{refs}

\begin{thebibliography}{10}

\bibitem{angelini2016low}
Patrizio Angelini, Michael~A Bekos, Till Bruckdorfer, Jaroslav Han{\v{c}}l,
  Michael Kaufmann, Stephen Kobourov, Antonios Symvonis, and Pavel Valtr.
\newblock Low ply drawings of trees.
\newblock In {\em International Symposium on Graph Drawing and Network
  Visualization}, pages 236--248. Springer, 2016.

\bibitem{angelini2017vertex}
Patrizio Angelini, Steve Chaplick, Felice De~Luca, Jiri Fiala, Jan Hancl~Jr,
  Niklas Heinsohn, Michael Kaufmann, Stephen Kobourov, Jan Kratochvil, and
  Pavel Valtr.
\newblock On vertex-and empty-ply proximity drawings.
\newblock {\em arXiv preprint arXiv:1708.09233}, 2017.

\bibitem{breu1998unit}
Heinz Breu and David~G Kirkpatrick.
\newblock Unit disk graph recognition is {NP}-hard.
\newblock {\em Computational Geometry}, 9(1-2):3--24, 1998.

\bibitem{de2017experimental}
F.~De~Luca, E.~Di~Giacomo, W.~Didimo, S.~Kobourov, and G.~Liotta.
\newblock An experimental study on the ply number of straight-line drawings.
\newblock In {\em International Workshop on Algorithms and Computation (to
  appear)}. Springer, 2017.

\bibitem{di2015low}
Emilio Di~Giacomo, Walter Didimo, Seok-hee Hong, Michael Kaufmann, Stephen~G
  Kobourov, Giuseppe Liotta, Kazuo Misue, Antonios Symvonis, and Hsu-Chun Yen.
\newblock Low ply graph drawing.
\newblock In {\em Information, Intelligence, Systems and Applications (IISA),
  2015 6th International Conference on}, pages 1--6. IEEE, 2015.

\bibitem{JGAA-251}
{Christian}~A. {Duncan}, {David} {Eppstein}, {Michael}~T. {Goodrich},
  {Stephen}~G. {Kobourov}, and {Martin} {Nöllenburg}.
\newblock Lombardi drawings of graphs.
\newblock {\em Journal of Graph Algorithms and Applications}, 16(1):85--108,
  2012.

\bibitem{eppstein2008studying}
David Eppstein and Michael~T Goodrich.
\newblock Studying (non-planar) road networks through an algorithmic lens.
\newblock In {\em Proceedings of the 16th ACM SIGSPATIAL International
  Conference on Advances in Geographic Information Systems}, page~16. ACM,
  2008.

\bibitem{falconer2004fractal}
Kenneth Falconer.
\newblock {\em Fractal Geometry: Mathematical Foundations and Applications}.
\newblock John Wiley \& Sons, 2004.

\bibitem{Frati2008}
Fabrizio Frati and Maurizio Patrignani.
\newblock A note on minimum-area straight-line drawings of planar graphs.
\newblock In Seok-Hee Hong, Takao Nishizeki, and Wu~Quan, editors, {\em 15th
  Int. Symp. on Graph Drawing (GD)}, pages 339--344, 2008.

\bibitem{gansner2010gmap}
Emden~R Gansner, Yifan Hu, and Stephen Kobourov.
\newblock Gmap: Visualizing graphs and clusters as maps.
\newblock In {\em Visualization Symposium (PacificVis), 2010 IEEE Pacific},
  pages 201--208. IEEE, 2010.

\bibitem{gansner2004graph}
Emden~R Gansner, Yehuda Koren, and Stephen North.
\newblock Graph drawing by stress majorization.
\newblock In {\em International Symposium on Graph Drawing, LNCS 3383}, pages
  239--250. Springer, 2004.

\bibitem{Garg1994}
Ashim Garg and Roberto Tamassia.
\newblock Planar drawings and angular resolution: Algorithms and bounds.
\newblock In Jan van Leeuwen, editor, {\em 2nd European Symp. on Algorithms
  (ESA)}, pages 12--23, 1994.

\bibitem{hachul2004drawing}
Stefan Hachul and Michael J{\"u}nger.
\newblock Drawing large graphs with a potential-field-based multilevel
  algorithm.
\newblock In {\em International Symposium on Graph Drawing, LNCS 3383}, pages
  285--295. Springer, 2004.

\bibitem{hlinveny1995contact}
Petr Hlin{\v{e}}n{\`y}.
\newblock Contact graphs of curves.
\newblock In {\em International Symposium on Graph Drawing}, pages 312--323.
  Springer, 1995.

\bibitem{hlinveny1998classes}
Petr Hlin{\v{e}}n{\`y}.
\newblock Classes and recognition of curve contact graphs.
\newblock {\em Journal of Combinatorial Theory, Series B}, 74(1):87--103, 1998.

\bibitem{kamada1989algorithm}
Tomihisa Kamada and Satoru Kawai.
\newblock An algorithm for drawing general undirected graphs.
\newblock {\em Information Processing Letters}, 31(1):7--15, 1989.

\bibitem{Miller-1137}
Gary~L. Miller, Shang-Hua Teng, William Thurston, and Stephen~A. Vavasis.
\newblock Geometric separators for finite-element meshes.
\newblock {\em SIAM Journal on Scientific Computing}, 19(2):364--386, 1998.

\bibitem{mondal2017new}
Debajyoti Mondal and Lev Nachmanson.
\newblock A new approach to graphmaps, a system browsing large graphs as
  interactive maps.
\newblock {\em arXiv preprint arXiv:1705.05479}, 2017.

\bibitem{nachmanson2015graphmaps}
Lev Nachmanson, Roman Prutkin, Bongshin Lee, Nathalie~Henry Riche, Alexander~E
  Holroyd, and Xiaoji Chen.
\newblock Graphmaps: Browsing large graphs as interactive maps.
\newblock In {\em International Symposium on Graph Drawing and Network
  Visualization}, pages 3--15. Springer, 2015.

\bibitem{sleator1983data}
Daniel~D Sleator and Robert~Endre Tarjan.
\newblock A data structure for dynamic trees.
\newblock {\em Journal of Computer and System Sciences}, 26(3):362--391, 1983.

\end{thebibliography}

\clearpage
\section*{Appendix}
Here we include a proof of Lemma~\ref{lem:drawinglemma}.

\drawinglemma*
\begin{proof}
We prove each part of our lemma as follows.
\begin{enumerate}
\item Suppose that our heavy path decomposition tree has a total height of $H$, and the path $P$ is at height $h$. Then we use the \textsc{DrawPath} algorithm to construct a drawing of $P$. We set the drawing disk for a vertex $v_i$ in $P$ to have radius $n_i$, that is, the size of the subtrees anchored at $v_i$. Since the length of the edge $(v_i, v_{i + 1})$ is at least $n_i + n_{i + 1}$ (by Lemma~\ref{lem:two-drawing-length}), the drawing disks for any two adjacent vertices in our path will not overlap.

\item Next we scale the drawing of $P$ by $3^{\Delta (H - h)}$. Note that each path anchored at a vertex in $P$ is scaled by $3^{\Delta (H - (h + 1))}$, so the difference in the scaling factor is $3^{\Delta}$. We show that at least $\Delta - 1$ paths can be anchored in different layers around each vertex $v$ in $P$.

From Lemma~\ref{lem:two-drawing-length}, we know that each path anchored at $v$ has an unscaled length of at most $6n$, where $n$ is the total size of the subtrees anchored at $v$. We also know by Lemma~\ref{lemma:layer-size} that the ply disks for vertices in two different paths will not overlap if their distance from $v$ differs by at least a factor of three.

So we will draw the $j$th path anchored at $v_i$ is drawn between $x_j$ and $x_{j + 1}$, where $x_j$ satisfies the following recurrence:
\begin{align*}
x_1 &= 6n_i \\
x_i &= 3x_{i - 1} + 6n_i
\end{align*}

Solving the recurrence, we find that $x_j = 3n (3^j - 1)$. Since we have at most $\Delta - 1$ layers, the largest layer will have an outer radius less than $3^{\Delta} n_i$. Since the unscaled drawing disk for $v_i$ had a radius of $n_i$, a relative scaling factor of $3^{\Delta}$ is sufficient to fit the paths that are anchored at it.

\item Since our heavy path decomposition has height $O(\log n)$, the largest path will be scaled by a factor of $3^{O(\Delta \log n)}$ from its original length of $O(n)$. So the diameter of our drawing is $3^{O(\Delta \log n)} n$, which simplifies to $n^{O(\Delta)}$. The total area is then also $n^{O(\Delta)}$, for a larger constant.
\end{enumerate}
\end{proof}

\end{document}